\documentclass[10pt,conference]{IEEEtran}
\usepackage{amsmath,amssymb}
\usepackage{graphicx}
\usepackage{mathrsfs}
\usepackage{tikz}
\usepackage{booktabs}
\usepackage{floatflt}
\usepackage{enumerate}
\usepackage{psfrag}	
\usepackage{array}
\usepackage{multirow,hhline}
\usepackage{exscale}
\usepackage{color}
\usepackage{colortbl}
\usepackage{epsfig}
\usepackage{cite}
\usepackage{amsthm}

\renewcommand{\vec}[1]{\ensuremath{\mathbf{#1}}}

\newtheorem{theorem}{Theorem}
\newtheorem{lemma}{Lemma}
\newtheorem{remark}{Remark}

\newtheorem{definition}{Definition}

\begin{document}



\IEEEoverridecommandlockouts
\title{On the Capacity of the 2-user Gaussian MAC Interfering with a P2P Link}
\author{
\IEEEauthorblockN{Anas Chaaban and Aydin Sezgin}
\IEEEauthorblockA{Emmy-Noether Research Group on Wireless Networks\\
Institute of Telecommunications and Applied Information Theory\\
Ulm University, 89081 Ulm, Germany\\
Email: {anas.chaaban@uni-ulm.de, aydin.sezgin@uni-ulm.de}}
\thanks{%
This work is supported by the German Research Foundation, Deutsche
Forschungsgemeinschaft (DFG), Germany, under grant SE 1697/3.%
}
}

\maketitle


\begin{abstract}
A multiple access channel and a point-to-point channel sharing the same medium for communications are considered. We obtain an outer bound for the capacity region of this setup, and we show that this outer bound is achievable in some cases. These cases are mainly when interference is strong or very strong. A sum capacity upper bound is also obtained, which is nearly tight if the interference power at the receivers is low. In this case, using Gaussian codes and treating interference as noise achieves a sum rate close to the upper bound. 
\end{abstract}

\begin{IEEEkeywords}
Gaussian MAC, capacity region, strong interference, very strong interference. 
\end{IEEEkeywords}

\section{Introduction}
The multiple access channel (MAC) is a scenario where several transmitters want to deliver one message each to one receiver. The capacity region of the MAC is known since 1971 \cite{Ahlswede,Liao}. This setup models users that want to communicate with a base station in cell for example. 

The interference channel (IC) is another well studied model in information theory, that is a model where two transmit-receive pairs want to communicate while causing interference to each other. The capacity of the IC is known in some cases. For instance, the capacity region of the Gaussian IC with very-strong interference regime was obtained by \cite{Carleial_vsi}, its capacity region with strong interference was obtained by \cite{Sato}, and its sum capacity with noisy interference was obtained by \cite{AnnapureddyVeeravalli,ShangKramerChen,MotahariKhandani}.

We merge a MAC and an IC into one setup, by adding one more transmit-receive pair to the communications medium of a 2-user MAC. Then, we have a system with three transmitters and two receivers. The obtained setup consists of a point-to-point (P2P) channel interfering with with a 2-user MAC and we call it a (PIMAC). We study this model and obtain capacity results in some parameter ranges.

We first derive an outer bound on the capacity region of the PIMAC. This outer bound is obtained by using a technique similar to Sato's technique in \cite{Sato}. Then we provide an inner bound on the capacity region of the PIMAC, which is obtained by allowing both receivers to decode all messages in a MAC fashion. The given outer and inner bounds are shown to coincide in some cases. Namely, when the PIMAC has strong or very strong interference.

Moreover, similar to \cite{Carleial_vsi}, a regime where interference does not reduce capacity is obtained. In this case, the interference free capacity region of the MAC as well as the capacity of the P2P channel can be achieved.

The simple scheme of treating interference as noise at each receiver gives a sum capacity lower bound for the PIMAC. Using a genie aided approach similar to \cite{AnnapureddyVeeravalli}, we obtain a sum capacity upper bound which is very close to the lower bound of treating interference as noise if the interference power is low.

The rest of the paper is organized as follows. In Section \ref{Model}, the PIMAC is introduced. The main outer and inner bounds are given in Section \ref{Bounds}. Cases where the outer and inner bounds coincide are given in Section \ref{PIMAC_SI}. The PIMAC with weak interference is considered in Section \ref{WeakInterference} where sum capacity upper and lower bounds are given. Finally, we conclude with Section \ref{conclusion}.

\section{System Model}
\label{Model}

A Gaussian 2-user multiple access channel (MAC) and a Gaussian point-to-point (P2P) channel use the same medium for their communication. They interfere with each other as shown in Figure \ref{pIMAC}. We call the resulting setup the PIMAC, whose received signals can be written as
\begin{align}
Y_1&=X_1+X_2+h_{31}X_3+Z_1,\\
Y_2&=h_{12}X_1+h_{22}X_2+X_3+Z_2,
\end{align}
where $h_{ij}\in\mathbb{R}$ denotes the channel coefficient from transmitter $i\in\{1,2,3\}$ to receiver $j\in\{1,2\}$, $X_i$ is the transmit symbol of transmitter $i$ and $Z_j$ is an additive white Gaussian noise (AWGN) with $Z_j\sim\mathcal{N}(0,1)$. Transmitter $i$ has power constraint $P_i$ so that $\mathbb{E}[X_i^2]\leq P_i$.

\begin{figure}[h]
\centering
\includegraphics[width=0.8\columnwidth]{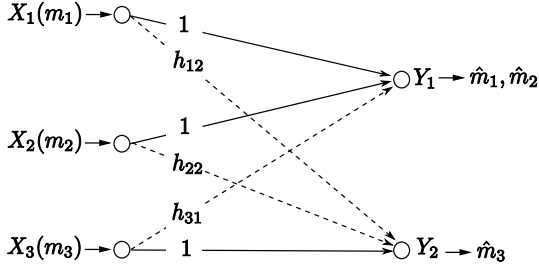}
\caption{Multiple Access Channel interfering with a point-to-point channel.}
\label{pIMAC}
\end{figure}
 
The transmitters encode their messages $m_i$ chosen independently from $\mathcal{M}_i=\{1,\dots,2^{nR_i}\}$ into a codeword  of length $n$ symbols $X_i^n\in\mathbb{R}^n$ $$X_i^n=(X_{i,1},\dots,X_{i,n})$$ using encoding functions $f_i$. Receiver 1 decodes $m_1$ and $m_2$ by using a decoding function $g_1$, i.e. $(\hat{m}_1,\hat{m}_2)=g(Y_1^n)$, and receiver 2 decodes $m_3$ by using a decoding function $g_2$, i.e. $\hat{m}_3=g_2(Y_2^n)$. A decoding error $E_i$ occurs if $m_i\neq\hat{m}_i$. 

An $(n,2^{nR_1},2^{nR_2},2^{nR_3})$ code for the PIMAC consists of encoding functions, decoding functions, and message sets, and induces an average probability of decoding error given by 
\begin{align}
P^{(n)}=\frac{1}{2^{nR_\Sigma}}\sum_{\vec{m}\in\mathcal{M}_1\times\mathcal{M}_2\times\mathcal{M}_3}P\left(\bigcup_{i=1}^3E_i\right),
\end{align}
where $R_\Sigma=\sum_{i=1}^3R_i$ and $\vec{m}=(m_1,m_2,m_3)$. A rate tuple $(R_1,R_2,R_3)$ is achievable if there exists a sequence of $(n,2^{nR_1},2^{nR_2},2^{nR_3})$ codes such that $P^{(n)}]\to0$ as $n\to\infty$. The closure of the set of all achievable rates is the capacity region of the PIMAC denoted $\mathcal{C}$, and the sum capacity is the highest achievable sum-rate 
\begin{align}
C_\Sigma=\max_{(R_1,R_2,R_3)\in\mathcal{C}}R_\Sigma.
\end{align}

In the next section we present an outer bound and an inner bound on the capacity region of the PIMAC, which we later show to be tight in some parameter regimes.

\section{Outer and Inner Bounds}
\label{Bounds}

We start by considering the PIMAC with strong interference. We first use the following definition for notational convenience.
\begin{definition}
We denote a multiple access channel MAC from transmitters $i\in\mathcal{S}\subseteq\{1,2,3\}$ to receiver $j\in\{1,2\}$ with AWGN with variance $1$ by $M(\mathcal{S},j)$ and its corresponding capacity region by $\mathcal{C}^M(\mathcal{S},j)$. This region is given by
\begin{equation*}
\mathcal{C}^M(\mathcal{S},j)=\left\{
\begin{array}{l}
R_i\geq0\ \ \forall i \in \mathcal{S}: \forall \mathcal{T}\subseteq\mathcal{S}\\
\sum_{i\in\mathcal{T}}R_i\leq\frac{1}{2}\log\left(1+\sum_{i\in\mathcal{T}}h_{ij}^2P_i\right)
\end{array}
\right\}.
\end{equation*}
\end{definition}
Now, we are ready to state the following lemma that will serve as an outer bound for $\mathcal{C}$.

\begin{lemma}
\label{pSILemma}
The capacity region of the PIMAC is outer bounded by $\overline{\mathcal{C}}$, $$\mathcal{C}\subseteq\overline{\mathcal{C}}$$ where
\begin{align*}
\overline{\mathcal{C}}=\left\{
\begin{array}{lr}
(R_1,R_2,R_3)\in\mathbb{R}_+^3:&\\
R_3\leq\frac{1}{2}\log(1+P_3)&\\
(R_1,R_2)\in\mathcal{C}^M(\{1,2\},1)&\\
(R_1,R_3)\in\mathcal{C}_M(\{1,3\},2)& \text{ if } h_{12}^2\geq1\\
(R_2,R_3)\in\mathcal{C}_M(\{2,3\},2)& \text{ if } h_{22}^2\geq1\\
(R_1,R_2,R_3)\in\mathcal{C}_M(\{1,2,3\},1)& \text{ if } h_{31}^2\geq1
\end{array}\right\}.
\end{align*} 
\end{lemma}
\begin{proof}
The proof is based on arguments similar to Sato's arguments in \cite{Sato} and is given in Appendix \ref{pSILemmaProof}.
\end{proof}

This outer bound $\overline{\mathcal{C}}$ given in Lemma \ref{pSILemma} can be shown to be tight in some cases. For this purpose, we need an achievable rate region for the PIMAC, which we give in the following theorem.

\begin{theorem}
\label{InnerBound}
The capacity region of the PIMAC is inner bounded by $\underline{\mathcal{C}}$, $$\underline{\mathcal{C}}\subseteq\mathcal{C}$$
where
\begin{align}
\underline{\mathcal{C}}=\left\{
\begin{array}{l}
(R_1,R_2,R_3)\in\mathbb{R}_+^3:\\
(R_1,R_2,R_3)\in\mathcal{C}_M(\{1,2,3\},1)\\
(R_1,R_2,R_3)\in\mathcal{C}_M(\{1,2,3\},2)
\end{array}\right\}.
\end{align} 
\end{theorem}
\begin{proof}
Each receiver simply decodes all messages in a MAC fashion. The achievable rates must be in the intersection of $\mathcal{C}_M(\{1,2,3\},1)$ and $\mathcal{C}_M(\{1,2,3\},2)$ and the statement of the theorem follows.
\end{proof}

\section{The PIMAC with strong interference}
\label{PIMAC_SI}
In general, the bounds given by $\overline{\mathcal{C}}$ and $\underline{\mathcal{C}}$, given in Lemma \ref{pSILemma} and Theorem \ref{InnerBound} respectively, are not always tight. However, under some conditions, they coincide. These conditions are given in the following theorem.

\begin{theorem}
\label{pSITheorem}
If the PIMAC satisfies the following conditions
\begin{align}
\label{pSI1}
h_{12}^2,h_{22}^2,h_{31}^2&\geq1\\
\label{pSI2}
h_{12}^2P_1+h_{22}^2P_2+P_3&\geq P_1+P_2+h_{31}^2P_3,
\end{align}
then the outer bound $\overline{\mathcal{C}}$ and the inner bound $\underline{\mathcal{C}}$ coincide and hence, the capacity region of the PIMAC is
\begin{align}
\mathcal{C}=\overline{\mathcal{C}}=\underline{\mathcal{C}}.
\end{align}
\end{theorem}
\begin{proof}
By expanding the achievable region $\underline{\mathcal{C}}$, and using conditions (\ref{pSI1}) and (\ref{pSI2}), it can be shown that $\underline{\mathcal{C}}$ reduces to $\overline{\mathcal{C}}$ and the result follows. Details are given in Appendix \ref{pSIProof}. 
\end{proof}

In conclusion, if the PIMAC has strong interference, i.e. (\ref{pSI1}) is satisfied, and if (\ref{pSI2}) holds, then the capacity region of the PIMAC is the intersection of two MAC capacity regions, from all transmitters to each receiver. The capacity region of a PIMAC satisfying (\ref{pSI1}) and (\ref{pSI2}) is shown in Figure \ref{pIMAC_SI_Region}.
\begin{figure}
\centering
\psfragscanon
\psfrag{x}[t]{$R_1$}
\psfrag{y}[b]{$R_2$}
\psfrag{z}[b]{$R_3$}
\includegraphics[width=\columnwidth]{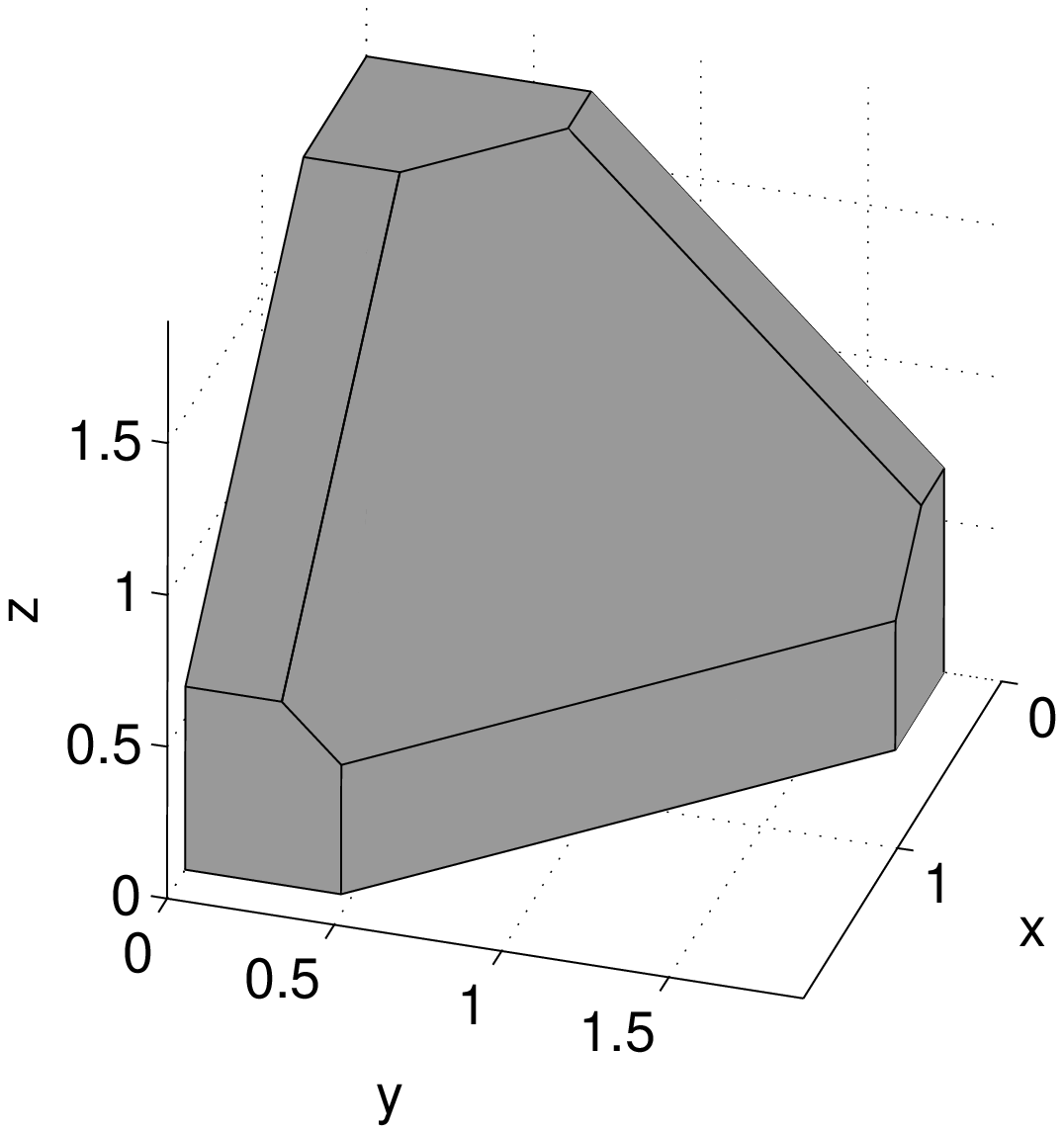}
\caption{The capacity region of a PIMAC satisfying conditions (\ref{pSI1}) and (\ref{pSI2}), with $P_1=P_2=P_3=10$, $h_{12}=1.2$, $h_{22}=1.5$ and $h_{31}=1.3$.}
\label{pIMAC_SI_Region}
\end{figure}

Let us now consider special cases of the PIMAC with strong interference $h_{12}^2,h_{22}^2,h_{31}^2\geq1$. According to Lemma \ref{pSILemma}, any achievable rate tuple $(R_1,R_2,R_3)\in\mathbb{R}_+^3$ for PIMAC with strong interference must satisfy
\begin{align}
\label{SIB3}
(R_1,R_3)&\in\mathcal{C}_M(\{1,3\},2)\\
\label{SIB4}
(R_2,R_3)&\in\mathcal{C}_M(\{2,3\},2)\\
\label{SIB5}
(R_1,R_2,R_3)&\in\mathcal{C}_M(\{1,2,3\},1).
\end{align}
where the remaining bounds $R_3\leq\frac{1}{2}\log(1+P_3)$ and $(R_1,R_2)\in\mathcal{C}^M(\{1,2\},1)$ were removed because they are already included in (\ref{SIB3})-(\ref{SIB5}).

\subsection{Very strong interference from transmitter 1 to receiver 2}
Consider a PIMAC with strong interference such that interference from transmitter 1 to receiver 2 is very strong, 
\begin{align}
\label{VSI12-1}
h_{12}^2&\geq1+P_3+h_{22}^2P_2\\
h_{22}^2&\geq1\\
\label{VSI12-3}
h_{31}^2&\geq1.
\end{align}
Since $h_{12}^2\geq1+P_3+h_{22}^2P_2\geq1+P_3$, then it can be shown that the bound (\ref{SIB3}) is redundant and can be removed. Thus the outer bound becomes
\begin{align}
\left\{
\begin{array}{l}
(R_1,R_2,R_3)\in\mathbb{R}_+^3:\\
(R_2,R_3)\in\mathcal{C}_M(\{2,3\},2)\\
(R_1,R_2,R_3)\in\mathcal{C}_M(\{1,2,3\},1)
\end{array}\right\}.
\end{align} 
Moreover $h_{12}^2\geq1+P_3+h_{22}^2P_2\Rightarrow$
\begin{align}
\frac{1}{2}\log\left(1+\frac{h_{12}^2P_1}{1+P_3+h_{22}^2P_2}\right)\geq\frac{1}{2}\log(1+P_1).
\end{align}
This means that the second receiver can decode $X_1^n$ while treating $X_2^n$ and $X_3^n$ as noise without imposing an additional rate constraint on the achievable rates. Then by subtracting the contribution of $X_1^n$ from $Y_2^n$, and decoding $X_2^n$ and $X_3^n$ in a MAC fashion at the second receiver, the rate region $\mathcal{C}^M(\{2,3\},2)$ can be achieved. If the first receiver decodes all signals in a MAC fashion, the rate region $\mathcal{C}_M(\{1,2,3\},1)$ can be achieved. Therefore we have the following theorem.
\begin{theorem}
If the PIMAC satisfies (\ref{VSI12-1})-(\ref{VSI12-3}) then its capacity region is
\begin{align}
\mathcal{C}=\left\{
\begin{array}{l}
(R_1,R_2,R_3)\in\mathbb{R}_+^3:\\
(R_2,R_3)\in\mathcal{C}_M(\{2,3\},2)\\
(R_1,R_2,R_3)\in\mathcal{C}_M(\{1,2,3\},1)
\end{array}\right\}.
\end{align} 
\end{theorem}
Similarly if interference from transmitter 2 to receiver 2 is very strong.

\subsection{Very strong interference from transmitter 3 to receiver 1}
Suppose that the PIMAC has
\begin{align}
\label{VSI31-1}
h_{12}^2&\geq1\\
h_{22}^2&\geq1\\
\label{VSI31-3}
h_{31}^2&\geq1+P_1+P_2.
\end{align}
In this case, the rate restriction (\ref{SIB5}) can be replaced with
\begin{align}
(R_1,R_2)&\in\mathcal{C}_M(\{1,2\},1),
\end{align}
which follows since the rate upper bounds on $R_3$, $R_1+R_3$, $R_2+R_3$, and $R_1+R_2+R_3$ are all redundant given $(R_1,R_3)\in\mathcal{C}_M(\{1,3\},2)$ and $(R_2,R_3)\in\mathcal{C}_M(\{2,3\},2)$. Therefore, the following capacity outer bound holds
\begin{align}
\left\{
\begin{array}{l}
(R_1,R_2,R_3)\in\mathbb{R}_+^3:\\
(R_1,R_3)\in\mathcal{C}_M(\{1,3\},2)\\
(R_2,R_3)\in\mathcal{C}_M(\{2,3\},2)\\
(R_1,R_2)\in\mathcal{C}_M(\{1,2\},1)
\end{array}\right\}.
\end{align}

This outer bound can be achieved if we add one more condition on the PIMAC as we show next. If (\ref{VSI31-3}) holds then receiver 1 can decode $X_3^n$ while treating $X_1^n$ and $X_2^n$ as noise without imposing an additional rate constraint on $m_3$. This follows since (\ref{VSI31-3}) $\Rightarrow$
\begin{align}
\frac{1}{2}\log\left(1+\frac{h_{31}^2P_3}{1+P_1+P_2}\right)\geq\frac{1}{2}\log(1+P_3).
\end{align}
Thus, receiver 1 can remove the contribution of $X_3^n$ from $Y_1^n$ and then decode $X_1^n$ and $X_2^n$ achieving $\mathcal{C}^M(\{1,2\},1)$. Let receiver 2 decode all signals in a MAC fashion achieving the region $(R_1,R_2,R_3)\in\mathcal{C}_M(\{1,2,3\},2)$. If we have
\begin{align}
\label{VSI31-4}
h_{12}^2P_1+h_{22}^2P_2\geq(P_1+P_2)(1+P_3),
\end{align}
then the bound on $R_1+R_2+R_3$ in $\mathcal{C}_M(\{1,2,3\},2)$ becomes redundant and can be removed. This follows since
\begin{align}
&h_{12}^2P_1+h_{22}^2P_2\geq(P_1+P_2)(1+P_3)\Rightarrow\\
&\frac{1}{2}\log(1+h_{12}^2P_1+h_{22}^2P_2+P_3)\geq\frac{1}{2}\log(1+P_1+P_2)\nonumber\\
&\hspace{5cm}+\frac{1}{2}\log(1+P_3).
\end{align}
Thus, under condition (\ref{VSI31-4}) and given $(R_1,R_2)\in\mathcal{C}^M(\{1,2\},1)$ we can replace $(R_1,R_2,R_3)\in\mathcal{C}_M(\{1,2,3\},2)$ by 
\begin{align}
(R_1,R_2)&\in\mathcal{C}_M(\{1,2\},2)\\
(R_1,R_3)&\in\mathcal{C}_M(\{1,3\},2)\\
(R_2,R_3)&\in\mathcal{C}_M(\{2,3\},2).
\end{align}
Now, since the region $\mathcal{C}_M(\{1,2\},2)$ is larger than $\mathcal{C}_M(\{1,2\},1)$, we obtain the following theorem.
\begin{theorem}
If the PIMAC satisfies (\ref{VSI31-1})-(\ref{VSI31-3}) and (\ref{VSI31-4}), then its capacity region is
\begin{align}
\mathcal{C}=\left\{
\begin{array}{l}
(R_1,R_2,R_3)\in\mathbb{R}_+^3:\\
(R_1,R_3)\in\mathcal{C}_M(\{1,3\},2)\\
(R_2,R_3)\in\mathcal{C}_M(\{2,3\},2)\\
(R_1,R_2)\in\mathcal{C}_M(\{1,2\},1)
\end{array}\right\}.
\end{align} 
\end{theorem}
This region is shown in Figure \ref{pIMAC_VSI31_Region} for a setup with $P_1=P_2=P_3=10$, $h_{12}=4.3$, $h_{22}=2$ and $h_{31}=4.6$.
\begin{figure}
\centering
\psfragscanon
\psfrag{x}[t]{$R_1$}
\psfrag{y}[b]{$R_2$}
\psfrag{z}[b]{$R_3$}
\includegraphics[width=\columnwidth]{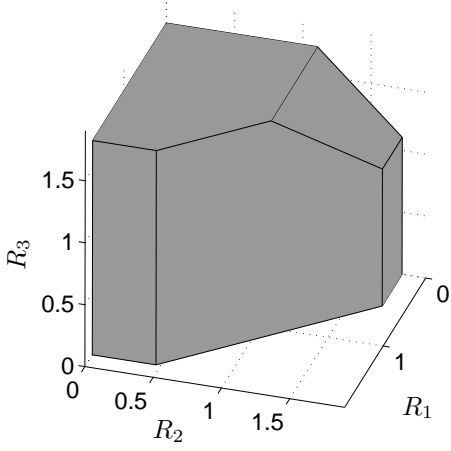}
\caption{The capacity region of a PIMAC satisfying conditions (\ref{VSI31-1})-(\ref{VSI31-3}) and (\ref{VSI31-4}), with $P_1=P_2=P_3=10$, $h_{12}=4.3$, $h_{22}=2$, and $h_{31}=4.6$.}
\label{pIMAC_VSI31_Region}
\end{figure}

\subsection{The PIMAC with very strong interference}
Consider now the PIMAC with very strong interference such that
\begin{align}
\label{pIMAC-VSI1}
h_{12}^2,h_{22}^2&\geq1+P_3\\
\label{pIMAC-VSI2}
h_{31}^2&\geq1+P_1+P_2.
\end{align}
In this case, interference does not reduce capacity, i.e. the capacity region of the interference free MAC $M(\{1,2\},1)$ can be achieved by $(R_1,R_2)$ and the point-to-point capacity $R_3=\frac{1}{2}\log(1+P_3)$ can also be achieved. 
Clearly, the following outer bound holds for the PIMAC
\begin{align}
\left\{
\begin{array}{l}
(R_1,R_2,R_3)\in\mathbb{R}_+^3:\\
R_3\leq\frac{1}{2}\log(1+P_3)\\
(R_1,R_2)\in\mathcal{C}_M(\{1,2\},1)
\end{array}\right\}.
\end{align}
This outer bound can be achieved under the given conditions (\ref{pIMAC-VSI1}) and (\ref{pIMAC-VSI2}) by decoding interference first at each receiver while treating the desired signal as noise, and then decoding the desired signal after removing the contribution of interference. Thus, as if interference is not present, the capacity region of the pIMAC becomes as given in the following theorem.
\begin{theorem}
If the PIMAC satisfies (\ref{pIMAC-VSI1}) and (\ref{pIMAC-VSI2}), then its capacity region is
\begin{align}
\mathcal{C}=\left\{
\begin{array}{l}
(R_1,R_2,R_3)\in\mathbb{R}_+^3:\\
R_3\leq\frac{1}{2}\log(1+P_3)\\
(R_1,R_2)\in\mathcal{C}_M(\{1,2\},1)
\end{array}\right\}.
\end{align} 
\end{theorem}
As example for this case, the capacity region of a PIMAC with very strong interference is shown in Figure \ref{pIMAC_VSI_Region}.
\begin{figure}
\centering
\psfragscanon
\psfrag{x}[t]{$R_1$}
\psfrag{y}[b]{$R_2$}
\psfrag{z}[b]{$R_3$}
\includegraphics[width=\columnwidth]{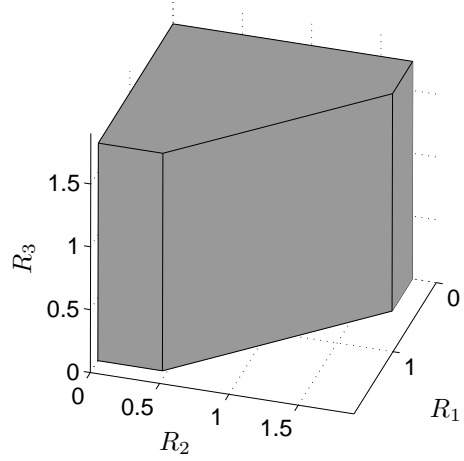}
\caption{The capacity region of a PIMAC satisfying conditions (\ref{pSI1}) and (\ref{pSI2}), with $P_1=P_2=P_3=10$, $h_{12}=h_{22}=3.5$, and $h_{31}=4.6$.}
\label{pIMAC_VSI_Region}
\end{figure}

\section{The PIMAC with weak interference}
\label{WeakInterference}
In \cite{AnnapureddyVeeravalli,ShangKramerChen,MotahariKhandani}, a genie aided technique was used to obtain a sum capacity upper bound for the IC that coincides with the simple lower bound of treating interference as noise. Thus the sum capacity of the IC in the so-called noisy interference regime was obtained. We use a similar technique to that used in \cite{AnnapureddyVeeravalli} to obtain an upper bound for the PIMAC. This upper bound is stated in the following theorem.

\begin{theorem}
\label{SumCapacityUpperBound}
The sum capacity of the PIMAC is upper bounded by $\overline{C}_\Sigma$, i.e. $$C_\Sigma\leq\overline{C}_\Sigma$$ where
\begin{align}
\overline{C}_\Sigma&\triangleq\min_{\substack{|\rho_1|,|\rho_2|\leq1,\\
\eta_1^2\leq1-\rho_2^2,\\
\eta_2^2\leq1-\rho_1^2}} I(X_{1G},X_{2G};Y_{1G},S_{1G})+I(X_{3G};Y_{2G},S_{2G})
\end{align}
with $X_{iG}\sim\mathcal{N}(0,P_i)$, $Y_{iG}$ is the corresponding channel output when the input is $X_{iG}$ and 
\begin{align}
S_{1G}&=h_{12}X_{1G}+h_{22}X_{2G}+\eta_1W_1\\
S_{2G}&=h_{31}X_{3G}+\eta_2W_2,
\end{align}
where $W_i\sim\mathcal{N}(0,1)$ for $i\in\{1,2\}$ such that $\mathbb{E}[W_iZ_i]=\rho_i$.
\end{theorem}
\begin{proof}
We use a genie aided approach similar to \cite{AnnapureddyVeeravalli} to bound $R_1+R_2$ and $R_3$. After adding the bounds, we observe that their sum is maximized by Gaussian inputs and we obtain the upper bound. Details are given in Appendix \ref{SumCapacityUpperBoundProof}.
\end{proof}

A trivial sum capacity lower bound is obtained by using Gaussian codes and treating interference as noise. The resulting lower bound is given by
\begin{align}
\label{CS_LB}
C_\Sigma\geq\underline{C}_\Sigma&=\frac{1}{2}\log\left(1+\frac{P_1+P_2}{1+h_{31}^2P_3}\right)\nonumber\\
&\quad+\frac{1}{2}\log\left(1+\frac{P_3}{1+h_{12}^2P_1+h_{22}^2P_2}\right).
\end{align}

In Figure \ref{Figure:SumCapacity}, we plot the upper bound $\overline{C}_\Sigma$ and the lower bound $\underline{C}_\Sigma$ as a function of SNR, the ration between the transmit power and the noise variance, for an IMAC with $P_1=P_2=P_3=P$, $h_{12}=h_{31}=0.2$, and $h_{22}=0.1$. Notice that this upper bound is nearly tight up to some value of SNR. Intuitively, this means that below some threshold value of the interference power, treating interference as noise achieves sum rate very close to the sum capacity $C_\Sigma$.

\begin{figure}
\centering
\psfragscanon
\psfrag{UB}[l]{$\overline{C}_\Sigma$}
\psfrag{LB}[l]{$\underline{C}_\Sigma$}
\psfrag{x}[t]{SNR(dB)}
\psfrag{y}[b]{Sum Rate(bits/channel use)}
\includegraphics[width=0.9\columnwidth]{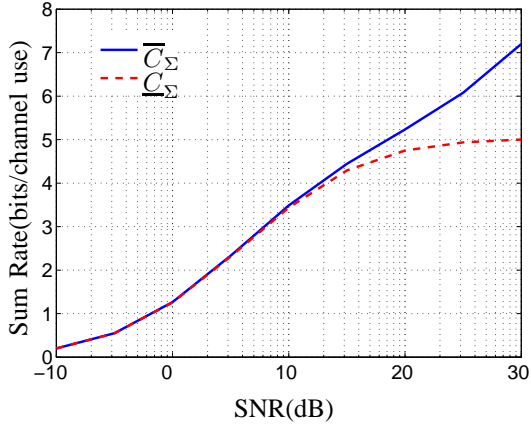}
\caption{Sum Capacity upper bound and lower bound for an IMAC with $P_1=P_2=P_3=P$, $h_{12}=h_{31}=0.2$, and $h_{22}=0.1$.}
\label{Figure:SumCapacity}
\end{figure}

\begin{remark}
In the special case of $h_{12}=h_{22}=h$, then the PIMAC becomes similar to a 2-user IC with transmit signals $X_1^n+X_2^n$ and $X_3^n$. In this case, the bounds $\overline{C}_\Sigma$ in Theorem (\ref{SumCapacityUpperBound}) and $\underline{C}_\Sigma$ in (\ref{CS_LB}) coincide if $|h|(1+h_{31}^2P_3)+|h_{31}|(1+h^2(P_1+P_2))\leq1$.
\end{remark}

\section{Conclusion}
\label{conclusion}
We studied a setup where a 2-user Gaussian multiple access channels interferes with a point-to-point Gaussian channel. We studied this setup under strong and very strong interference. We obtained the capacity region for some parameter regions. We also derived a sum capacity upper bound, which is very close to a sum capacity lower bound obtained by using Gaussian codes and treating interference as noise if the interference power is low.

\bibliography{/home/chaaban/tex/myBib}

\begin{appendices}

\section{Proof of Lemma \ref{pSILemma}}
\label{pSILemmaProof}
Clearly, for any achievable rate triple $(R_1,R_2,R_3)$, the rate $R_3$ is upper bounded by the interference free capacity of the P2P channel, and the rate couple $(R_1,R_2)$ is included in the  capacity region of the interference free MAC $M(\{1,2\},1)$, i.e. 
\begin{align}
R_3&\leq\frac{1}{2}\log(1+P_3)\\
(R_1,R_2)&\in\mathcal{C}^{M}(\{1,2\},1).
\end{align}

Now we derive the other bounds. A genie gives $m_2$ as additional information to the second receiver. The obtained genie aided channel has a larger capacity region than the original PIMAC and can be used to outer bound $\mathcal{C}$. Now consider a rate tuple $(R_1,R_2,R_3)$ in the capacity region of the genie aided channel. This means that the first receiver and the second receiver are able to decode $(m_1,m_2)$ and $m_3$ reliably respectively. Since the second receiver is able to decode $m_3$, and since it knows $m_2$ from the genie, then it is able to construct 
\begin{align}
\tilde{Y}_1^n&=\frac{Y_2^n-h_{22}X_2^n-X_3^n}{h_{12}}+X_2^n+h_{31}X_3^n\nonumber\\
&=X_1^n+X_2^n+h_{31}X_3^n+\frac{1}{h_{12}}Z_2^n
\end{align}
If $h_{12}^2\geq1$, then $\tilde{Y}_1^n$ is a less noisy version of $Y_1^n$. So if the first receiver is able to decode $m_1$ then so does the second receiver. Thus
$(R_1,R_3)$ is contained in the MAC region of $M(\{2,3\},2)$, i.e. 
\begin{align}
(R_1,R_3)\in\mathcal{C}^{M}(\{1,3\},2)  \text{ if } h_{12}^2\geq1.
\end{align}
Similarly, we can show that
\begin{align}
(R_2,R_3)\in\mathcal{C}^{M}(\{2,3\},2) &\text{ if } h_{22}^2\geq1.
\end{align}
Finally, for any rate triple $(R_1,R_2,R_3)\in\mathcal{C}$, the first receiver is able to decode $m_1$ and $m_2$ reliably. Thus, if $h_{31}^2\geq1$, receiver 1 can construct a less noisy version of $Y_2^n$ and then decode $m_3$ reliably leading to the bound
\begin{align}
(R_1,R_2,R_3)\in\mathcal{C}^{M}(\{1,2,3\},1) &\text{ if } h_{31}^2\geq1.
\end{align}
and the statement of the Lemma is proved.

\section{Proof of Theorem \ref{pSITheorem}}
\label{pSIProof}
We can expand the achievable region $\underline{\mathcal{C}}$
to obtain the following rate constraints for the achievable region
\begin{align*}
R_1&\leq\frac{1}{2}\log(1+\min\{P_1,h_{12}^2P_1\})\\
R_2&\leq\frac{1}{2}\log(1+\min\{P_2,h_{22}^2P_2\})\\
R_3&\leq\frac{1}{2}\log(1+\min\{P_3,h_{31}^2P_3\})\\
R_1+R_2&\leq\frac{1}{2}\log(1+\min\{P_1+P_2,h_{12}^2P_1+h_{22}^2P_2\})\\
R_1+R_3&\leq\frac{1}{2}\log(1+\min\{P_1+h_{31}^2P_3,h_{12}^2P_1+P_3\})\\
R_2+R_3&\leq\frac{1}{2}\log(1+\min\{P_2+h_{31}^2P_3,h_{22}^2P_2+P_3\})\\
R_1+R_2+R_3&\leq\frac{1}{2}\log(1+\min\{P_1+P_2+h_{31}^2P_3,\\
&\hspace{3cm}h_{12}^2P_1+h_{22}^2P_2+P_3\})
\end{align*}
Using conditions (\ref{pSI1}) and (\ref{pSI2}), these constraints reduce to
\begin{align*}
R_i&\leq\frac{1}{2}\log(1+P_i),\ \forall i\in\{1,2,3\}\\
R_1+R_2&\leq\frac{1}{2}\log(1+P_1+P_2)\\
R_1+R_3&\leq\frac{1}{2}\log(1+P_1+h_{31}^2P_3)\\
R_1+R_3&\leq\frac{1}{2}\log(1+h_{12}^2P_1+P_3)\\
R_2+R_3&\leq\frac{1}{2}\log(1+P_2+h_{31}^2P_3)\\
R_2+R_3&\leq\frac{1}{2}\log(1+h_{22}^2P_2+P_3)\\
R_1+R_2+R_3&\leq\frac{1}{2}\log(1+P_1+P_2+h_{31}^2P_3).
\end{align*}
Expanding $\overline{\mathcal{C}}$ when conditions (\ref{pSI1}) and (\ref{pSI2}) hold lead to the same rate constraints. Thus, the two regions $\overline{\mathcal{C}}$ and $\underline{\mathcal{C}}$ coincide and the result follows.

\section{Proof of Theorem \ref{SumCapacityUpperBound}}
\label{SumCapacityUpperBoundProof}
Consider the following signals
\begin{align}
S_1^n&=h_{12}X_1^n+h_{22}X_2^n+\eta_1W_1^n\\
S_2^n&=h_{31}X_3^n+\eta_2W_2^n
\end{align}
where $\eta_i\in\mathbb{R}$ and $W_i^n$ is an i.i.d. sequence with $W_i\sim\mathcal{N}(0,1)$ for $i\in\{1,2\}$ such that $\mathbb{E}[W_iZ_i]=\rho_i$. A genie gives these signals as extra information to the receivers. Thus, receiver $i$ knows both $Y_i^n$ and $S_i^n$ and uses them for decoding. Using Fano's inequality, we have
\begin{align}
n(R_1+R_2)&\leq I(X_1^n,X_2^n;Y_1^n,S_1^n)+n\epsilon_{1n}\\
nR_3&\leq I(X_3^n;Y_2^n,S_2^n)+n\epsilon_{2n}
\end{align}
where $\epsilon_{1n},\epsilon_{2n}\to0$ as $n\to\infty$. Using the chain rule, we can write
\begin{align}
n(R_1+R_2-\epsilon_{1n})&\leq I(X_1^n,X_2^n;S_1^n)+I(X_1^n,X_2^n;Y_1^n|S_1^n)\nonumber\\
&=h(S_1^n)-h(S_1^n|X_1^n,X_2^n)+h(Y_1^n|S_1^n)\nonumber\\
\label{Q1}
&\quad-h(Y_1^n|S_1^n,X_1^n,X_2^n).
\end{align}
Similarly
\begin{align}
n(R_3-\epsilon_{2n})&\leq h(S_2^n)-h(S_2^n|X_3^n)+h(Y_2^n|S_2^n)\nonumber\\
\label{Q2}
&\quad-h(Y_2^n|S_2^n,X_3^n).
\end{align}
The next step is to show that sum of the above quantities (\ref{Q1}) and (\ref{Q2}) are maximized when $X_i^n$ is i.i.d. Gaussian for all $i\in\{1,2,3\}$, i.e. $X_i\sim\mathcal{N}(0,P_i^*)$ where $P_i^*\leq P_i$. We will denote that Gaussian $X_i$ as $X_{iG}$ and the corresponding outputs as $Y_{iG}$ and $S_{iG}$. Since the Gaussian distribution maximizes the conditional differential entropy under a covariance constraint, we have
\begin{align}
h(Y_1^n|S_1^n)&\leq nh(Y_{1G}|S_{1G})\\
h(Y_2^n|S_2^n)&\leq nh(Y_{2G}|S_{2G}).
\end{align}
Furthermore, the terms $h(S_1^n|X_1^n,X_2^n)$ and $h(S_2^n|X_3^n)$ are independent of the distribution of $X_i$. Now, we jointly maximize the difference $h(S_1^n)-h(Y_2^n|S_2^n,X_3^n)$ as follows
\begin{align}
&\hspace{-1cm}h(S_1^n)-h(Y_2^n|S_2^n,X_3^n)\nonumber\\
&=h(h_{12}X_1^n+h_{22}X_2^n+\eta_1W_1^n)\nonumber\\
&\quad-h(h_{12}X_1^n+h_{22}X_2^n+Z_2^n|W_2^n)\\
&\stackrel{(a)}{=}h(h_{12}X_1^n+h_{22}X_2^n+\eta_1W_1^n)\nonumber\\
&\quad-h(h_{12}X_1^n+h_{22}X_2^n+V_2^n)\\
&\stackrel{(b)}{\leq}nh(h_{12}X_{1G}+h_{22}X_{2G}+\eta_1W_1)\nonumber\\
&\quad-nh(h_{12}X_{1G}+h_{22}X_{2G}+V_2)\\
\label{T1}
&=nh(S_{1G})-nh(Y_{2G}|S_{2G},X_{3G},X_{4G}),
\end{align}
where $(a)$ follows from \cite[Lemma 6]{AnnapureddyVeeravalli} with $V_2^n$ being an i.i.d. Gaussian sequence with $V_2\sim\mathcal{N}(0,1-\rho_2^2)$ and $(b)$ from the worst case noise \cite{DiggaviCover} such that $\eta_1^2\leq 1-\rho_2^2$. Since $\eta_1^2\leq1-\rho_2^2$, (\ref{T1}) can be shown to be increasing in $P_1^*$ and $P_2^*$ and thus is maximized by $X_{iG}\sim\mathcal{N}(0,P_i)$. Similarly, if $\eta_2^2\leq1-\rho_1^2$, $h(S_2^n)-h(Y_1^n|S_1^n,X_1^n,X_2^n)$ can be shown to be maximized by $X_{iG}\sim\mathcal{N}(0,P_i)$. Moreover, $h(Y_{1G}|S_{1G})$ and $h(Y_{2G}|S_{2G})$ can be shown to be increasing in $P_i^*$ and thus also maximized by $X_{iG}\sim\mathcal{N}(0,P_i)$. As a result, by letting $n\to\infty$
\begin{align}
\label{SumRate_n}
&R_1+R_2+R_3\nonumber\\
&\quad\leq I(X_{1G},X_{2G};Y_{1G},S_{1G})+I(X_{3G};Y_{2G},S_{2G})
\end{align}
with $X_{iG}\sim\mathcal{N}(0,P_i)$ and $\eta_i\leq1-\rho_j^2$.

\end{appendices}

\end{document}